\documentclass[a4paper,twoside, 11pt]{journal}
\usepackage{fullpage}
\usepackage{amssymb, amsmath, amsthm, wasysym, mathrsfs}
\usepackage{xspace,xcolor}
\usepackage{graphicx, wrapfig}
\usepackage{plaatjes}

\usepackage[colorlinks]{hyperref}
\usepackage[english]{babel}
\usepackage{tabularx}
\usepackage{titlesec}
\usepackage{booktabs}
\usepackage{comment}
\usepackage{times}
\usepackage[scaled]{helvet}
\usepackage{microtype}
\usepackage{enumitem}

\usepackage[left=1in,top=1in,right=1in,bottom=1in]{geometry}
\graphicspath{{figures/}}

\hypersetup{allcolors = violet}

\author
{
  Marc van Kreveld\thanks
  {
    Department of Information and Computing Sciences, Universiteit Utrecht,
    The Netherlands,
    \texttt{\{m.j.vankreveld|m.loffler|f.staals\}@uu.nl}
  }
  \and Maarten L\"offler\footnotemark[1]
  \and Frank Staals\footnotemark[1]
}
\title{Central Trajectories}

\newcommand{\remark}[3]{\textcolor{blue}{\textsc{#1 #2:}} \textcolor{red}{\textsf{#3}}}
\renewcommand{\remark}[3]{}
\newcommand{\frank}[2][says]{\remark{Frank}{#1}{#2}}

\newcommand{\maarten}[2][says]{\remark{Maarten}{#1}{#2}}

\newtheorem{theorem} {Theorem}
\newtheorem{lemma}[theorem]{Lemma}
\newtheorem{corollary}[theorem] {Corollary}

\newtheorem{observation}[theorem] {Observation}

\renewenvironment{proof}{\noindent\textit{Proof.}}{\hfill$\qed$\smallskip}

\newcommand{\mkmcal}[1]{\ensuremath{\mathcal{#1}}\xspace}

\newcommand{\G}{\mkmcal{G}}
\newcommand{\F}{\mkmcal{F}}

\newcommand{\C}{\mkmcal{C}}
\newcommand{\X}{\mkmcal{X}}
\newcommand{\D}{\mkmcal{D}}

\newcommand{\A}{\mkmcal{A}}
\newcommand{\RG}{\mkmcal{R}}

\newcommand{\mkmbb}[1]{\ensuremath{\mathbb{#1}}\xspace}

\newcommand{\R}{\mkmbb{R}}

\newcommand{\eps}{\varepsilon}

\newcommand{\etal}{et al.\xspace}
\newcommand{\dd}[1]{\ensuremath{\,\mathrm{d}#1}}

\DeclareFontFamily{OT1}{pzc}{}
\DeclareFontShape{OT1}{pzc}{m}{it}{<-> s * [1] pzcmi7t}{}
\DeclareMathAlphabet{\mathpzc}{OT1}{pzc}{m}{it}
\newcommand{\mkpzc}[1]{\ensuremath{\mathpzc{#1}}\xspace}

\renewcommand{\c}{\mkpzc{C}}
\renewcommand{\t}{\mkpzc{T}}

\newcommand{\U}{\mkpzc{U}}
\renewcommand{\L}{\mkpzc{L}}

\newcommand{\I}{\mkpzc{I}}

\newcommand{\trajectoid}{trajectoid\xspace}

\makeatletter
\renewcommand*{\@fnsymbol}[1]{\ensuremath{\ifcase#1\or *\or \mathsection\or \mathparagraph\or
   \dagger\or \ddagger\or \|\or **\or \dagger\dagger
   \or \ddagger\ddagger \else\@ctrerr\fi}}
\makeatother

\titleformat{\paragraph}[runin]{\bfseries}{\theparagraph}{0}{}[.]
\titlespacing{\paragraph}{%
  0pt}{%              left margin
  0.5\baselineskip}{% space before (vertical)
  1em}%               space after (horizontal)

\begin{document}
\maketitle

\begin{abstract}
  An important task in trajectory analysis is clustering. The results of a
  clustering are often summarized by a single representative trajectory and
  an associated size of each
  cluster. We study the problem of computing a suitable representative of a
  set of similar trajectories. To this
  end we define a \emph{central trajectory} \c, which consists of pieces of the
  input trajectories, switches from one entity to another only if they are
  within a small distance of each other, and such that at any time $t$, the
  point $\c(t)$ is as central as possible. We measure centrality in terms of
  the radius of the smallest disk centered at $\c(t)$ enclosing all entities at
  time $t$, and discuss how the techniques can be adapted to other measures of
  centrality. We first study the problem in $\R^1$, where we show that an
  optimal central trajectory \c representing $n$ trajectories, each consisting
  of $\tau$ edges, has complexity $\Theta(\tau n^2)$ and can be computed in
  $O(\tau n^2 \log n)$ time. We then consider trajectories in $\R^d$ with $d\geq 2$,
  and show that the complexity of \c is at most $O(\tau n^{5/2})$ and can be computed in
  $O(\tau n^3)$ time.
\end{abstract}

% keywords: moving entities, trajectories, median, combinatorial complexity,
%           algorithms, computational geometry, Reeb-graph

%\rule{\linewidth}{0.75pt}
\thispagestyle{empty}
\clearpage
\setcounter{page}{1}

\section{Introduction}
\label{sec:Introduction}

A \emph {trajectory} is a sequence of time-stamped locations in the plane, or more generally in $\R^d$.
Trajectory data is obtained by tracking the movements of e.g. animals \cite{BovetB88,Calenge200934,gal-nmibc-09}, hurricanes \cite{Stohl1998947}, traffic \cite{lltx-dftf-10}, or other moving entities \cite{dwf-rpm-09} over time.
Large amounts of such data have recently been collected in a variety of research fields.
As a result, there is a great demand for tools and
techniques to analyze trajectory data.

One important task in trajectory analysis is \emph {clustering}: subdividing a large collection of trajectories into groups of ``similar'' ones. This problem has been studied
extensively, and many different techniques are available~\cite{bbgll-dcpcs-11,grsc-pcecu-07,gs-tcmrm-99,lhw-tc-07,vgk-dsmt-02}.
Once a suitable clustering has been determined, the result needs to be stored
or prepared for further processing. Storing the whole collection of
trajectories in each cluster is often not feasible, because follow-up analysis
tasks may be computation-intensive. Instead, we wish to represent each cluster
by a signature: the number of trajectories in the cluster, together
with a \emph {representative} trajectory which should capture the defining
features of all trajectories in the cluster.

Representative trajectories are also useful for visualization
purposes. Displaying large amounts of trajectories often leads to visual
clutter. Instead, if we show only a number of representative trajectories, this
reduces the visual clutter, and allows for more effective data
exploration. The original trajectories can still be shown if desired,
using the details-on-demand principle in information visualization~\cite{Shneiderman96}.
%\frank{We may want to say a bit more here.}

\paragraph {Representative trajectories}
When choosing a representative trajectory for a group of similar trajectories,
the first obvious choice would be to pick one of the trajectories in the
group. However, one can argue that no single element in a group may be a
good representative, e.g. because each individual trajectory has some prominent
feature that is not shared by the rest (see Fig.~\ref {fig:single_bad}), or no
trajectory is sufficiently in the middle all the time. On the
other hand, it is desirable to output a trajectory that does consist of \emph
{pieces} of input trajectories, because otherwise the representative trajectory
may display behaviour that is not present in the input, e.g. because of
contextual information that is not available to the algorithm (see Fig.~\ref
{fig:lake}).

\tweeplaatjes {single_bad} {lake}
{ (a) Every trajectory has a peculiarity that is not representative for the set.
  (b) Taking, for example, the pointwise average of a set of trajectories may result in one that ignores context.
}

To determine what a good representative trajectory of a group of similar trajectories is, we identify two main categories: \emph {time-dependent} and \emph {time-independent} representatives.
Trajectories are typically collected as a discrete sequence of time-stamped
locations.  By linearly interpolating the locations we obtain a continuous
piecewise-linear curve as the image of the function. Depending on the application, we may be interested in the curve with attached time stamps (say, when studying a flock of animals that moved together) or in just the curve (say, when considering hikers that took the same route, but possibly at different times and speeds).

When time is not important, one can select a representative based directly on the geometry or
topology of the set of curves~\cite{bbklsww-mt-12,hpr-fdre-11}.  When time is
important, we would like to have the property that at each time $t$ our
representative point $c(t)$ is a good representative of the set of points
$P(t)$.  To this end, we may choose any static representative point of a point
set, for which many examples are available: the Fermat-Weber point (which
minimizes the sum of distances to the points in $P$), the center of mass (which
minimizes the sum of squared distances), or the center of the smallest
enclosing circle (which minimizes the distance to the farthest point in $P$).

%\maarten {Up to here I'm reasonably happy with the introduction.}
\paragraph {Central trajectories}

In this work, we focus on time-dependent measures based on static concepts of centrality. We choose the distance to the farthest point, but discuss in Section~\ref {sec:Extensions} how our results can be adapted to other measures.
%\maarten {Make sure that we really do this.}

Ideally, we would output a trajectory \c such that at any time $t$, $\c(t)$ is
the point (entity) that is closest to its farthest entity. Unfortunately, when
the entities move in $\R^d$ for $d > 1$, this may cause discontinuities. Such
discontinuities are unavoidable: if we insist that the output trajectory
consists of pieces of input trajectories \emph {and} is continuous, then in
general, there will be no opportunities to switch from one trajectory to
another, and we are effectively choosing one of the input trajectories
again. At the same time, we do not want to output a trajectory with arbitrarily
large discontinuities. An acceptable compromise is to allow discontinuities, or
\emph{jumps}, but only over small distances, controlled by a parameter
$\eps$. We note that this problem of discontinuities shows up for
time-independent representatives for entities moving in $\R^d$, with $d \geq
3$, as well, because the traversed curves generally do not intersect.

\paragraph{Related work}
Buchin \etal~\cite{bbklsww-mt-12} consider the problem of computing a \emph {median} trajectory for a set of trajectories without time information. Their method produces a trajectory that consists of pieces of the input.
Agarwal~\etal~\cite{agarwal2005staying} consider trajectories with time
information and compute a representative trajectory that follows the median (in $\R^1$) or a point of high \emph {depth} (in $\R^2$) of the input entities.
The resulting trajectory does not necessarily stay close to the input trajectories.
They give exact and approximate algorithms.
%Kinetic medians: \cite{agarwal2002kinetic}
%
Durocher and Kirkpatrick~\cite{durocher2009projection} observe that a trajectory minimizing the sum of distances to the other entities is \emph {unstable}, in the sense that arbitrarily small movement of the entities may cause an arbitrarily large movement in the location of the representative entity.
They proceed to consider alternative measures of centrality, and define the \emph {projection median}, which they prove is more stable.
Basu~\etal~\cite{basu2012projection} extend this concept to higher dimensions.

\paragraph{Problem description} We are given a set \X of $n$ entities, each
moving along a piecewise linear trajectory in $\R^d$ consisting of $\tau$
edges. We assume that all trajectories have their vertices at the same times,
i.e.~times $t_0,..,t_\tau$.  Fig.~\ref {fig:intro_example_standardandtime}
shows an example.
%
%\vierplaatjes [scale=.75] {intro_example} {intro_example_time} {intro_example_distances} {intro_example_slabs}
\tweeplaatjes [scale=1.18] {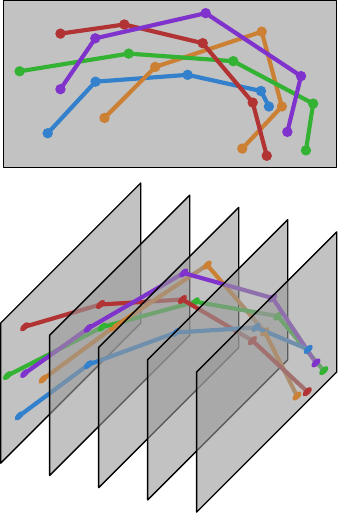}
{intro_example_slabsanddistances} { (a) Two views of five moving entities and
  their trajectories. (b) On the top the pairwise distances between the
  entities as a function over time. On the bottom the functions $D_\sigma$, and
  in yellow the area representing $\D(\c)$. }
For an entity $\sigma$, let $\sigma(t)$ denote the position of $\sigma$ at
time~$t$. With slight abuse of notation we will write $\sigma$ for both entity
$\sigma$ and its trajectory. At a given time $t$, we denote the distance from
$\sigma$ to the entity farthest away from $\sigma$ by $D_\sigma(t) =
D(\sigma,t) = \max_{\psi \in \X} \|\sigma(t)\psi(t)\|$, where $\|pq\|$ denotes
the Euclidean distance between points $p$ and $q$ in $\R^d$.
Fig.~\ref{fig:intro_example_slabsanddistances} illustrates the pairwise
distances and resulting $D$ functions for five example trajectories.  For ease
of exposition, we assume that the trajectories are in general position: that is, no
three trajectories intersect in the same point, and no two pairs of entities
are at distance $\eps$ from each other at the same time.

A \emph{\trajectoid} is a function that maps time to the set of entities \X,
with the restriction that at discontinuities the distance between the entities
involved is at most $\eps$. Intuitively, a \trajectoid corresponds to a
concatenation of pieces of the input trajectories in such a way that two
consecutive pieces match up in time, and the end point of the former piece is
within distance $\eps$ from the start point of the latter piece. In Fig.~\ref
{fig:intro_example_slabsanddistances}, a \trajectoid may switch between a pair
of entities when their pairwise distance function lies in the bottom strip of
height $\eps$. More formally, for a \trajectoid \t we have that

\begin{itemize}[nosep]
\item at any time $t$, $\t(t) = \sigma$ for some $\sigma \in \X$, and
\item at every time $t$ where \t has a discontinuity, that is, \t \emph{jumps}
  from entity $\sigma$ to entity $\psi$, we have that $\|\sigma(t)\psi(t)\|
  \leq \eps$.
\end{itemize}

Note that this definition still allows for a series of jumps within an
arbitrarily short time interval $[t,t+\delta]$, essentially simulating a jump
over distances larger than $\eps$. To make the formulation cleaner, we slightly
weaken the second condition, and allow a trajectoid to have discontinuities
with a distance larger than $\eps$, provided that such a large jump can be
realized by a sequence of small jumps, each of distance at most $\eps$. When it
is clear from the context, we will write $\t(t)$ instead of $\t(t)(t)$ to mean
the location of entity $\t(t)$ at time $t$. We now wish to compute a trajectoid
\c that minimizes the function
  \[ \D(\t) = \int_{t_0}^{t_\tau} D(\t,t) \dd t. \]
\noindent
So, at any time $t$, all entities lie in a disk of radius $D(\c,t)$ centered at
$\c(t)$.

\paragraph{Outline and results} We first study the situation where entities
move in $\R^1$. In Section~\ref{sec:oned} we show that the worst case
complexity of a central trajectory in $\R^1$ is $\Theta(\tau n^2)$, and that we
can compute one in $O(\tau n^2 \log n)$ time. We then extend our approach to
entities moving in $\R^d$, for any constant $d$, in Section~\ref
{sec:higher_dimensions}. For this case, we prove that the maximal complexity of
a central trajectory \c is $O(\tau n^{5/2})$. Computing \c takes $O(\tau n^3)$
time and requires $O(\tau n^2 \log n)$ working space. We briefly discuss various
extensions to our approach in Section~\ref{sec:Extensions}. Omitted proofs can
be found in Appendix~\ref{app:Omitted_Proofs}.

Even though we do not expect this to happen in practice, the worst case
complexity of our central trajectories can be significantly higher than the input size.
If this occurs, we can use traditional line simplification algorithms like Imai and
Iri~\cite{imai1998computational} to simplify the resulting central
trajectory. This gives us a representative that still is always close
---for instance within distance $2\eps$--- to one of the input trajectories.
Alternatively, we can use dynamic-programming combined with our methods
to enforce the output trajectory to have at most $k$ vertices, for any $k$,
and always be on the input trajectories.
Computing such a central trajectory is more expensive
than our current algorithms, however. Furthermore, enforcing a low output
complexity may not be necessary. For example, in applications like
visualization, the number of trajectories shown often has a larger impact
visual clutter than the length or complexity of the individual trajectories. It
may be easier to follow a single trajectory that has many vertices than
to follow many trajectories that have fewer vertices each.

% However, we expect that in most use cases $\tau$, the number
% of vertices in each trajectory, will be much larger than $n$, making the
% dependency on $\tau$ more important. Our central trajectories have linear
% complexity in $\tau$.

\section{Entities moving in $\R^1$}
\label{sec:oned}

\tweeplaatjes [scale=.8] {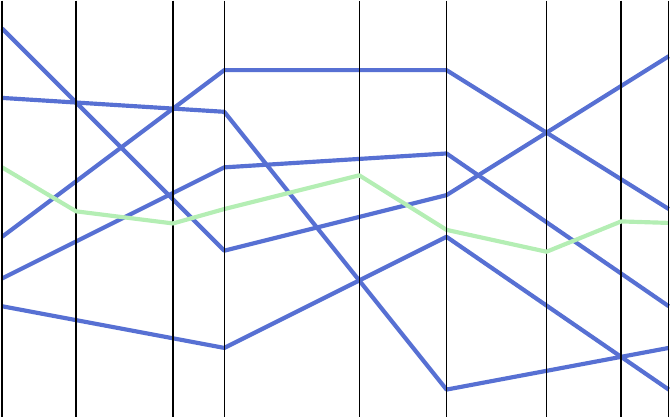} {1d_slabs_straight} {(a) A set of trajectories and the
  ideal trajectory \I. The breakpoints in the ideal trajectory partition time
  into $O(n\tau)$ intervals. (b) The trajectories after transforming \I to a
  horizontal line.}

Let \X be the set of entities moving in $\R^1$. The trajectories of these
entities can be seen as polylines in $\R^2$: we associate time with the
horizontal axis, and $\R^1$ with the vertical axis (see
Fig.~\ref{fig:1d_slabs}). We observe that the distance between two points $p$
and $q$ in $\R^1$ is simply their absolute difference, that is, $\|pq\|=|p-q|$.

% \begin{figure}[t]
%   \centering
%   \includegraphics{line_arrangement}
%   \caption{A central trajectory \c (blue) in arrangement \A, and starting
%     point $p$. \frank{TODO:update fig}}
%   \label{fig:line_arrangement}
% \end{figure}

Let \I be the \emph{ideal} trajectory, that is, the trajectory that minimizes
\D but is not restricted to lie on the input trajectories. It follows that at
any time $t$, $\I(t)$ is simply the average of the highest entity $\U(t)$ and
the lowest entity $\L(t)$. We further subdivide each time interval
$J_i=[t_i,t_{i+1}]$ into \emph{elementary intervals}, such that \I is a single
line segment inside each elementary interval.

\begin{lemma}
  \label{lem:num_elementary_intervals}
  The total number of elementary intervals is $\tau(n+2)$.
\end{lemma}

\begin{proof}
  The ideal trajectory \I changes direction when $\U(t)$ or $\L(t)$
  changes. During a single interval $[t_i,t_{i+1}]$ all entities move along lines, so \U
  and \L are the upper and lower envelope of a set of $n$ lines. So by standard
  point-line duality, \U and \L correspond to the upper and lower hull of $n$
  points. The summed complexity of the upper and lower hull is at most $n+2$.
\end{proof}

We assume without loss of generality that within each elementary interval \I
coincides with the $x$-axis. To simplify the description of the proofs and
algorithms, we also assume that the entities never move parallel to the ideal
trajectory, that is, there are no horizontal edges.

\begin{lemma}
  \label{lem:central_ideal}
  \c is a central trajectory in $\R^1$ if and only if it minimizes the function
  \[ \D'(\t) = \int_{t_0}^{t_\tau} |\t(t)| \dd t. \]
\end{lemma}

\begin{proof}
  A central trajectory \c is a \trajectoid that minimizes the function

  \begin{align*}
    \D(\t) &= \int_{t_0}^{t_\tau} D(\t,t) \dd t
            = \int_{t_0}^{t_\tau} \max_{\psi \in \X} \|\t(t)\psi(t)\|  \dd t
            = \int_{t_0}^{t_\tau} \max_{\psi \in \X} |\t(t) - \psi(t)| \dd t \\
           &= \int_{t_0}^{t_\tau} \max \{ |\t(t) - \U(t)|, |\t(t) - \L(t)| \} \dd t.
  \end{align*}

  Since $(\U(t) + \L(t))/2 = 0$, we have that $|\t(t) - \U(t)| > |\t(t) - \L(t)|$
  if and only if $\t(t) < 0$. So, we split the integral, depending on $\t(t)$,
  giving us

 \begin{align*}
  \D(\t) &= \int_{t_0 \leq t \leq t_\tau \land \t(t) \geq 0} \t(t) - \L(t) \dd t +
            \int_{t_0 \leq t \leq t_\tau \land \t(t) <    0} \U(t) - \t(t) \dd t \\
  &=     \int_{t_0 \leq t \leq t_\tau \land \t(t) \geq 0}  \t(t) \dd t -
         \int_{t_0 \leq t \leq t_\tau \land \t(t) \geq 0} \L(t) \dd t~+ \\
  &\quad \int_{t_0 \leq t \leq t_\tau \land \t(t) <    0} \U(t)  \dd t -
         \int_{t_0 \leq t \leq t_\tau \land \t(t) <    0} \t(t)  \dd t.
 \end{align*}

 We now use that $-\int_{\t(t) < 0}\t(t) = \int_{\t(t) < 0} |\t(t)|$, and that
 $-\int \L(t) = \int \U(t)$ (since $(\U(t) + \L(t))/2 = 0$). After rearranging the
 terms we then obtain

 \begin{align*}
  \D(\t) &=     \int_{t_0 \leq t \leq t_\tau \land \t(t) \geq 0}  \t(t)  \dd t +
                \int_{t_0 \leq t \leq t_\tau \land \t(t) <    0} |\t(t)| \dd t~+  \\
         &\quad \int_{t_0 \leq t \leq t_\tau \land \t(t) \geq 0} \U(t)   \dd t +
                \int_{t_0 \leq t \leq t_\tau \land \t(t) <    0} \U(t)   \dd t \\
         &=     \int_{t_0 \leq t \leq t_\tau} |\t(t)|  \dd t~+
                \int_{t_0 \leq t \leq t_\tau} \U(t)   \dd t.
  \end{align*}

  The last term is independent of $\t$, so we have $\D(\t) = \D'(\t) + c$,
  for some $c \in \R$. The lemma follows.
\end{proof}

By Lemma~\ref{lem:central_ideal} a central trajectory \c is a \trajectoid that
minimizes the area $\D'(\t)$ between \t and the ideal trajectory \I. Hence, we
can focus on finding a trajectoid that minimizes $\D'$.

\subsection{Complexity}
\label{sub:complexity_1d}

\eenplaatje[scale=.8]{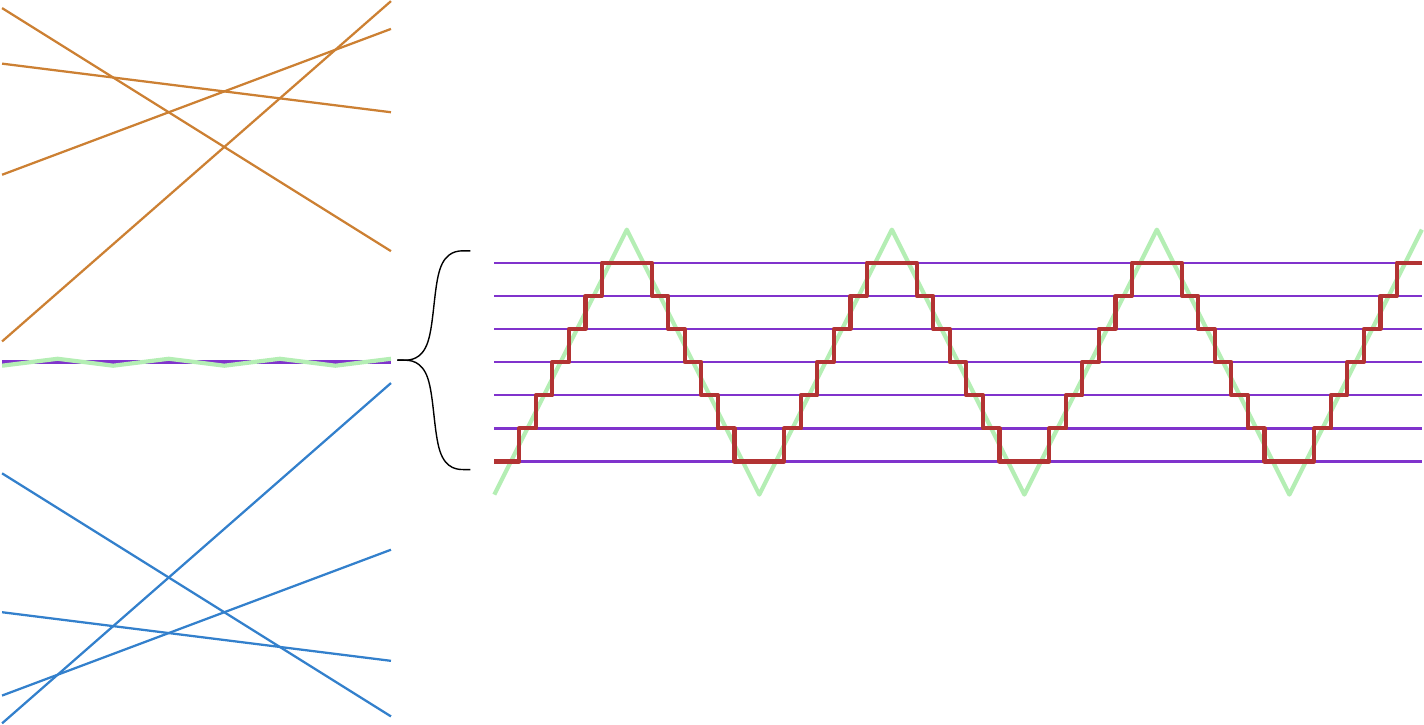}{Lower bound construction that shows that
  \C (red) may have quadratic complexity. The ideal trajectory \I is shown in green.}

\begin{lemma}
  \label{lem:lowerbound_complexity_1D}
  For a set of $n$ trajectories in $\R^1$, each with vertices at times
  $t_0,..,t_\tau$, a central trajectory \c may have worst case complexity
  $\Omega(\tau n^2)$.
\end{lemma}

\begin{proof}
  We describe a construction for the entities that shows that within a single
  time interval $J=[t_i,t_{i+1}]$ the complexity of \c may be
  $\Omega(n^2)$. Repeating this construction $\tau$ times gives us $\Omega(\tau
  n^2)$ as desired.

  Within $J$ the entities move linearly. So we construct an arrangement \A of
  lines that describes the motion of all entities. We place $m=n/3$ lines such
  that the upper envelope of \A has linear complexity. We do the same for the
  lower envelope. We position these lines such that the ideal trajectory \I
  ---which is the average of the upper and lower envelope--- makes a vertical
  ``zigzagging'' pattern (see
  Fig.~\ref{fig:quadratic_ideal}). %\frank{maybe explain how we achieve this}.
  The remaining set $H$ of $m$ lines are horizontal. Two consecutive lines are
  placed at (vertical) distance at most $\eps$. We place all lines such that
  they all intersect \I. It follows that \c jumps $O(n^2)$ times between the
  lines in $H$. The lemma follows.
\end{proof}

Two entities $\sigma$ and $\psi$ are $\eps$\emph{-connected} at time $t$ if
there is a sequence $\sigma=\sigma_0,..,\sigma_k=\psi$ of entities such that
for all $i$, $\sigma_i$ and $\sigma_{i+1}$ are within distance $\eps$ of each
other at time $t$. A subset $\X' \subseteq \X$ of entities is $\eps$-connected
at time $t$ if all entities in $\X'$ are pairwise $\eps$-connected at time
$t$. The set $\X'$ is $\eps$-connected during an interval $I$, if they are
$\eps$-connected at any time $t \in I$. We now observe:

\begin{observation}
  \label{obs:jump}\hspace{-3pt}
  \c can jump from entity $\sigma$ to $\psi$ at time $t$ if and only if
  $\sigma$ and $\psi$ are $\eps$-connected at time $t$.
\end{observation}

At any time $t$, we can partition \X into maximal sets of $\eps$-connected
entities. The central trajectory \c must be in one of such maximal sets $\X'$:
it uses the trajectory of an entity $\sigma \in \X'$ (at time $t$), if and only
if $\sigma$ is the entity from $\X'$ closest to \I. More formally,
let $f_\sigma(t) = |\sigma(t)|$, and let $\L(\F) = \min_{f \in \F} f$ denote
the lower envelope of a set of functions \F.

\begin{observation}
  \label{obs:lower_envelope_1d}
  Let $\X' \ni \sigma$ be a set of entities that is $\eps$-connected during
  interval $J$, and assume that $\c \in \X'$ during $J$. For any time $t \in
  J$, we have that $\c(t) = \sigma(t)$ if and only if $f_\sigma$ is on the
  lower envelope of the set $\F' = \{f_\psi \mid \psi \in \X'\}$ at time $t$,
  that is, $f_\sigma(t) = \L(\F)(t)$.
\end{observation}

Let $\X_1,..,\X_m$, denote a collection of maximal sets of entities that are
$\eps$-connected during time intervals $J_1,..,J_m$, respectively. Let $\F_i =
\{ f_\sigma \mid \sigma \in \X_i \}$, and let $\L_i$ be the lower envelope
$\L(\F_i)$ of $\F_i$ restricted to interval $J_i$. A lower envelope $\L_i$
has a break point at time $t$ if $f_\sigma(t) = f_\psi(t)$, for $\sigma,\psi
\in \X_i$. There are two types of break points: (i) $\sigma(t) = \psi(t)$,
or (ii) $\sigma(t) = -\psi(t)$. At events of type (i) the modified trajectories
of $\sigma$ and $\psi$ intersect. At events of the type (ii), $\sigma$ and
$\psi$ are equally far from \I, but on different sides of \I. Let $B = \{
(t,\sigma,\psi) \mid \L_i(t) = f_\sigma(t) = f_\psi(t) \land i \in \{1,..,m\}\}$
denote the collection of break points from all lower envelopes $\L_1,..,\L_m$.

\begin{lemma}
  \label{lem:break_points_in_at_most_one_set}
  Consider a triplet $(t,\sigma,\psi) \in B$. There is at most one lower
  envelope $\L_i$ such that $t$ is a break point in $\L_i$.
\end{lemma}

\begin{proof}
  Assume by contradiction that $t$ is a break point in both $\L_i$ and
  $\L_j$. At any time $t$, an entity can be in at most one maximal set
  $\X_\ell$. So if $\X_i$ and $\X_j$ share either entity $\sigma$ or $\psi$,
  then the intervals $J_i$ and $J_j$ are disjoint. It follows $t$ cannot lie
  in both intervals, and thus cannot be a break point in both $\L_i$ and
  $\L_j$. Contradiction.
\end{proof}

\begin{lemma}
  \label{lem:entities_on_boundary_zone}
  Let \A be an arrangement of $n$ lines, describing the movement of $n$
  entities during an elementary interval $J$. If there is a break point
  $(t,\sigma,\psi) \in B$, with $t \in J$, of type (ii), then $\sigma(t)$ and
  $\psi(t)$ lie on the boundary $\partial\mathcal{Z}$ of the zone $\mathcal{Z}$
  of \I in \A.
\end{lemma}

\begin{proof}
  Let $\X_j$ be the maximal $\eps$-connected set containing $\sigma$ and $\psi$, and
  assume without loss of generality that $f_\sigma(t) = \sigma(t) = -\psi(t) =
  f_\psi(t)$. Now, assume by contradiction that $\sigma$ is not on $\partial
  \mathcal{Z}$ at time $t$ (the case that $\psi(t)$ is not on
  $\partial\mathcal{Z}$ is symmetric). This means that there is an entity
  $\rho$ with $0 \leq \rho(t) < \sigma(t)$. If $\rho \in \X_j$, this
  contradicts that $f_\sigma(t)$ was on the lower envelope of $\X_j$ at time
  $t$. So $\rho$ is not $\eps$-connected to $\sigma$ at time $t$. Hence, their
  distance is at least $\eps$. We then have $\sigma(t) > \rho(t) + \eps >
  \eps$. It now follows that $\sigma$ and $\psi$ cannot be $\eps$-connected at
  time $t$: the distance between $\sigma$ and $\psi$ is bigger than $\eps$ so
  they are not directly connected, and $f_\sigma$ and $f_\psi$ are on $\L_j$,
  so there are also no other entities in $\X_j$ through which they can be
  $\eps$-connected. Contradiction.
\end{proof}

\begin{lemma}
  \label{lem:equidistant_jumps}
  Let \A be an arrangement of $n$ lines, describing the movement of $n$
  entities during an elementary interval $J$. The total number of break points
  $(t,\sigma,\phi) \in B$, with $t \in J$, of type (ii) is at most $6.5n$.
\end{lemma}

\begin{wrapfigure}[10]{r}{0.4\textwidth}
    \centering
    \vspace{-1em}
    \includegraphics{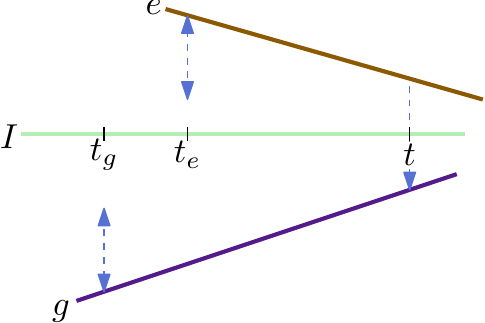}
    \vspace{-1em}
    \caption{The jumps of \L (dashed arrows) involving edges $e$ and $g$.}
    \label{fig:jumps_in_zone}
\end{wrapfigure}
\begin{proof}
  By Lemma~\ref{lem:break_points_in_at_most_one_set} all break points can be
  charged to exactly one set $\X_j$. From
  Lemma~\ref{lem:entities_on_boundary_zone} it follows that break points of
  type (ii) involve only entities whose lines in \A participate in the zone of
  \I.

  Let $E$ be the set of edges of $\partial\mathcal{Z}$. We have that $|E| \leq
  5.5n$~\cite{bepy-htlp-91,pa-cg-95}. We now split every edge that intersects
  \I, at the intersection point. Since every line intersects \I at most once,
  this means the number of edges in $E$ increases to $6.5n$. For every pair
  of edges $(e,g)$, that lie on opposite sides of \I, there is at most one time
  $t$ where a lower envelope $\L=\L_j$, for some $j$, has a break point of type
  (ii).

  Consider a break point of type (ii), that is, a time $t$ such that \L
  switches (jumps) from an entity $\sigma$ to an entity $\psi$, with $\sigma$
  and $\psi$ on opposite sides of \I. Let $e \in E$ and $g \in E$ be the
  edges containing $\sigma(t)$ and $\psi(t)$, respectively. If the arriving
  edge $g$ has not been charged before, we charge the jump to $g$. Otherwise,
  we charge it to $e$. We continue to show that every edge in $E$ is charged at most
  once. Since $E$ has at most $6.5n$ edges, the number of break points of type
  (ii) is also at most $6.5n$.

  We now show that either $e$ or $g$ has not been charged before. Assume, by
  contradiction, that both $e$ and $g$ have been charged before time $t$, at
  times $t_e$ and $t_g$, respectively. Consider the case that $t_g < t_e$
  (see Fig.~\ref{fig:jumps_in_zone}). At time $t_e$, the lower envelope \L
  jumps from an edge $h$ onto $e$ or vice versa. Since there is a jump
  involving edge $g$ at time $t_g$ and one at time $t$ it follows that at time
  $t_e$, $g$ is the closest edge in $E$ opposite to $e$. Hence, $h = g$. This
  means we jump twice between $e$ and $g$. Contradiction. The case $t_e < t_g$
  is symmetrical and the case $t_e = t_g$ cannot occur. It follows that $e$ or
  $g$ was not charged before time $t$, and thus all edges in $E$ are charged at
  most once.
\end{proof}

\begin{lemma}
  \label{lem:complexity_lower_envelopes_1d}
  The total complexity of all lower envelopes $\L_1,..,\L_m$ on $[t_i,t_{i+1}]$
  is $O(n^2)$.
\end{lemma}

\begin{proof}
  The break points in the lower envelopes are either of type (i) or of type
  (ii). We now show that there are at most $O(n^2)$ break points of either
  type.

  The break points of type (i) correspond to intersections between the
  trajectories of two entities. Within interval $[t_i,t_{i+1}]$ the entities
  move along lines, hence there are at most $O(n^2)$ such intersections. By
  Lemma~\ref{lem:break_points_in_at_most_one_set} all break points can be
  charged to exactly one set $\X_i$. It follows that the total number of break
  points of type (i) is $O(n^2)$.

  To show that the number of events of the second type is at most $O(n^2)$ as
  well we divide $[t_i,t_{i+1}]$ in $O(n)$ elementary intervals such that
  \I coincides with the $x$-axis. By Lemma~\ref{lem:equidistant_jumps} each such
  elementary interval contains at most $O(n)$ break points of type (ii).
\end{proof}

\begin{theorem}
  \label{thm:complexity_c_1D}
  Given a set of $n$ trajectories in $\R^1$, each with vertices at times
  $t_0,..,t_\tau$, a central trajectory \c has worst case complexity
  $O(\tau n^2)$.
\end{theorem}

\begin{proof}
  A central trajectory \c is a piecewise function. From
  Observations~\ref{obs:jump} and~\ref{obs:lower_envelope_1d} it now follows
  that \c has a break point at time $t$ only if (a) two subsets of entities become
  $\eps$-connected or $\eps$-disconnected, or (b) the lower envelope of a set
  of $\eps$-connected entities has a break point at time $t$. Within a single
  time interval $J_i=[t_i,t_{i+1}]$ there are at most $O(n^2)$ times when two
  entities are at distance exactly $\eps$. Hence, the number of events of type
  (a) during interval $J_i$ is also $O(n^2)$. By
  Lemma~\ref{lem:complexity_lower_envelopes_1d} the total complexity of all
  lower envelopes of $\eps$-connected sets during $J_i$ is also
  $O(n^2)$. Hence, the number of break points of type (b) within interval $J_i$
  is also $O(n^2)$. The theorem follows.
\end{proof}

\subsection{Algorithm}
\label{sub:Algorithm_1d}

\begin{wrapfigure}[14]{r}{0.45\textwidth}
    \centering
    \vspace{-3em}
    \includegraphics{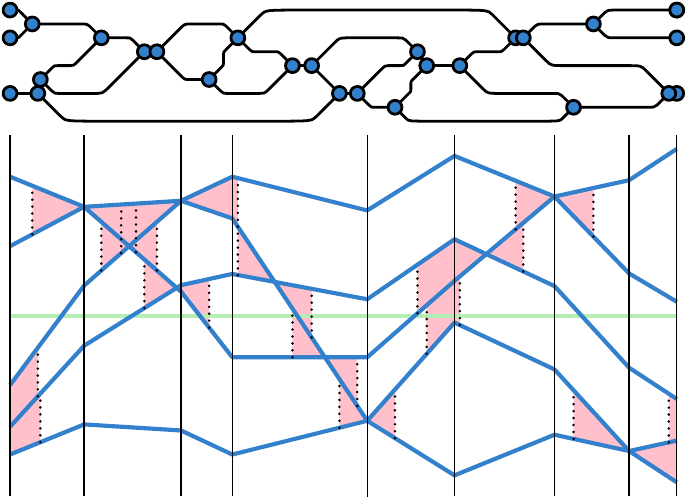}
    \caption{The Reeb graph for the moving entities from
  Fig.~\ref{fig:1d_slabs}. The dashed lines indicate that two entities are at
  distance $\eps$.}
    \label{fig:1d_reeb}
    \vspace{-1em}
\end{wrapfigure}
We now present an algorithm to compute a \trajectoid \c minimizing $\D'$. By
Lemma~\ref{lem:central_ideal} such a trajectoid is a central trajectory. The
basic idea is to construct a weighted (directed acyclic) graph that represents
a set of \trajectoid{}s containing an optimal \trajectoid. We can then find \c
by computing a minimum weight path in this graph.

The graph that we use is a weighted version of the Reeb graph that Buchin
\etal~\cite{grouping2013} use to model the trajectory grouping structure. We
review their definition here. The \emph{Reeb graph} \RG is a directed acyclic
graph. Each edge $e=(u,v)$ of \RG corresponds to a maximal subset of entities $C_e
\subseteq \X$ that is $\eps$-connected during the time interval
$[t_u,t_v]$. The vertices represent times at which the sets of $\eps$-connected
entities change, that is, the times at which two entities $\sigma$ and $\psi$
are at distance $\eps$ from each other and the set containing $\sigma$ merges with
or splits from the set containing $\psi$.  See Fig.~\ref{fig:1d_reeb} for an illustration.
% \frank{should we make a note that on  the grouping thing the RG is used for 2D traj.?}

By Observation~\ref{obs:jump} a central trajectory \c can jump from $\sigma$ to
$\psi$ if and only if $\sigma$ and $\psi$ are $\eps$-connected, that is, if
$\sigma$ and $\psi$ are in the same component $C_e$ of edge $e$. From
Observation~\ref{obs:lower_envelope_1d} it follows that on each edge $e$, \c
uses only the trajectories of entities $\sigma$ for which $f_\sigma$ occurs on
the lower envelope of the functions $\F_e = \{ f_\sigma \mid \sigma \in
C_e\}$. Hence, we can then express the cost for \c using edge $e$ by
\[
\omega_e = \int_{t_u}^{t_v} \L(\F_e)(t) \dd t.
\]

It now follows that \c follows a path in the Reeb graph \RG, that is, the set
of trajectoids represented by \RG contains a trajectoid minimizing $\D'$. So we
can compute a central trajectory by finding a minimum weight path in \RG
from a source to a sink.

\paragraph{Analysis} First we compute the Reeb graph as defined by Buchin
\etal~\cite{grouping2013}. This takes $O(\tau n^2\log n)$ time. Second we
compute the weight $\omega_e$ for each edge $e$. The Reeb graph \RG is a DAG,
so once we have the edge weights, we can use dynamic programming to compute a
minimum weight path in $O(|\RG|) = O(\tau n^2)$ time.  So all that remains is
to compute the edge weights $\omega_e$. For this, we need the lower envelope
$\L_e$ of each set $\F_e$ on the interval $J_e$. To compute the lower
envelopes, we need the ideal trajectory \I, which we can compute \I in $O(\tau n\log
n)$ time by computing the lower and upper envelope of the trajectories in each
time interval $[t_i,t_{i+1}]$.

Lemma~\ref{lem:complexity_lower_envelopes_1d} implies that the total complexity
of all lower envelopes is $O(\tau n^2)$. To compute them we have two
options. We can simply compute the lower envelope from scratch for every edge
of \RG. This takes $O(\tau n^2 \cdot n\log n) = O(\tau n^3 \log n)$
time. Instead, for each time interval $J_i=[t_i,t_{i+1}]$, we compute the
arrangement \A representing the modified trajectories on the interval $J_i$,
and use it to trace $\L_e$ in \A for every edge $e$ of \RG.

An arrangement of $m$ line segments can be built in $O(m \log m + A)$ time,
where $A$ is the output complexity \cite{as-aa-00}. In total \A consists of
$O(n^2)$ line segments: $n+2$ per entity. Since each pair of trajectories
intersects at most once during $J_i$, we have that $A = O(n^2)$. Thus, we can
build \A in $O(n^2 \log n)$ time. The arrangement \A represents all break
points of type (i), of all functions $f_\sigma$. We now compute all pairs of
points in \A corresponding to break points of type (ii). We do this in $O(n^2)$
time by traversing the zone of \I in \A.

We now trace the lower envelopes through \A: for each edge $e=(u,v)$ in the
Reeb graph with $J_e \subseteq J_i$, we start at the point $\sigma(t_u)$,
$\sigma \in C_e$, that is closest to \I, and then follow the edges in \A
corresponding to $\L_e$, taking care to jump when we encounter break points of
type (ii). Our lower envelopes are all disjoint (except at endpoints), so we
traverse each edge in \A at most once. The same holds for the jumps. We can
avoid costs for searching for the starting point of each lower envelope by
tracing the lower envelopes in the right order: when we are done tracing
$\L_e$, with $e=(u,v)$, we continue with the lower envelope of an outgoing edge
of vertex $v$. If $v$ is a split vertex where $\sigma$ and $\psi$ are at
distance $\eps$, then the starting point of the lower envelope of the other edge is
either $\sigma(t_v)$ or $\psi(t_v)$, depending on which of the two is farthest
from \I. It follows that when we have \A and the list of break points of type
(ii), we can compute all lower envelopes in $O(n^2)$ time. We conclude:

\begin{theorem}
  \label{thm:central_trajectory_1d}
  Given a set of $n$ trajectories in $\R^1$, each with vertices at times
  $t_0,..,t_\tau$, we can compute a central trajectory \c in $O(\tau n^2\log n)$ time
  using $O(\tau n^2)$ space.
\end{theorem}

\paragraph{A central trajectory without jumps} When our entities move in
$\R^1$, it is not yet necessary to have discontinuities in \c, i.e.~we can set
$\eps = 0$. In this case we can give a more precise bound on the complexity of
\c, and we can use a slightly easier algorithm. The details can be found in
Appendix~\ref{app:oned_no_jumps}.

\section{Entities moving in $\R^d$}
\label{sec:higher_dimensions}

In the previous section, we used the ideal trajectory \I, which minimizes the
distance to the farthest entity, ignoring the requirement to stay on an input
trajectory. The problem was then equivalent to finding a \trajectoid that
minimizes the distance to the ideal trajectory. In $\R^d$, with $d > 1$,
however, this approach fails, as the following example shows.

\eenplaatje {2d_not_ideal} {Point $p$ is closest to the ideal point $m$,
  however the smallest enclosing disk centered at $q$ is smaller than that of $p$.}

\begin {observation}
  Let $P$ be a set of points in $\R^2$.  The point in $P$ that minimizes the distance to
  the ideal point (i.e., the center of the smallest enclosing disk of $P$) is
  not necessarily the same as the point in $P$ that minimizes the distance to
  the farthest point in $P$.
\end {observation}
\begin {proof}
  See Fig.~\ref {fig:2d_not_ideal}.  Consider three points $a$, $b$ and $c$
  at the corners of an equilateral triangle, and two points $p$ and $q$ close
  to the center $m$ of the circle through $a$, $b$ and $c$.  Now $p$ is closer
  to $m$ than $q$, yet $q$ is closer to $b$ than $p$ (and $q$ is as far from
  $a$ as from $b$).
\end {proof}

\subsection{Complexity}
\label{sub:Complexity_2D}

It follows from Lemma~\ref{lem:lowerbound_complexity_1D} that the complexity of
a central trajectory for entities moving in $\R^d$ is at least $\Omega(\tau
n^2)$. In this section, we prove that the complexity of \c within a single time
interval $[t_i,t_{i+1}]$ is at most $O(n^{5/2})$. Thus, the complexity over all
$\tau$ time intervals is $O(\tau n^{5/2})$.

Let \F denote the collection of functions $D_\sigma$, for $\sigma \in \X$. We
partition time into intervals $J'_1,..,J'_{k'}$ such that in each interval
$J'_i$ all functions $D_\sigma$ restricted to $J'_i$ are \emph{simple}, that
is, they consist of just one piece. We now show that each function $D_\sigma$
consists of at most $\tau(2n-1)$ pieces, and thus the total number of intervals
is at most $O(\tau n^2)$. See Fig.~\ref{fig:intro_example_slabsanddistances}
for an illustration.

\begin{lemma}
  \label{lem:D_hyperbolic}
  Each function $D_\sigma$ is piecewise hyperbolic and
  consists of at most $\tau(2n-1)$ pieces.
\end{lemma}
%\frank{$D_\sigma$ is even convex, since the upper envelope of convex total
%  functions is again convex. Do we also want to mention that here?}

\begin{proof}
  Consider a time interval $J_i=[t_i,t_{i+1}]$. For any entity $\psi$ and any
  time $t \in J_i$, the function $\|\sigma(t)\psi(t)\| = \sqrt{at^2 + bt + c}$,
  with $a,b,c \in \R$, is hyperbolic in $t$. % \frank{actually, we even
    % have $a,c \in \R^{\geq 0}$}
  Each pair of such functions can intersect at most twice. During $J_i$,
  $D_\sigma$ is the upper envelope of these functions, so it consists of
  $\lambda_2(n)$ pieces, where $\lambda_s$ denotes the maximum complexity of
  a Davenport-Schinzel sequence of order $s$~\cite{as-dssga-00}. We have
  $\lambda_2(n) = 2n -1$, so the lemma follows.
\end{proof}

% \begin{figure}[t]
%   \centering
%   \includegraphics{graphs_D}
%   \caption{The functions $D(\sigma_i,\cdot)$ for 4 entities. The vertical grey
%     lines indicate the break points of one of the functions
%     $D(\sigma_i,\cdot)$. The vertical dashed lines indicate times where two
%     entities are at distance exactly $\eps$ from each other. \frank{the
%       fig. shows 4 pw hyperbolic functions, not necc. the functions $D_\sigma$. }}
%   \label{fig:graphs_D}
% \end{figure}

\begin{lemma}
  \label{lem:intersections_F}
  The total number of intersections of all functions in \F is at most $O(\tau n^3)$.
\end{lemma}

\begin{proof}
  Fix a pair of entities $\sigma,\psi$. By Lemma~\ref{lem:D_hyperbolic} there
  are at most $\tau(2n-1)$ time intervals $J$, such that $D_\sigma$ restricted
  to $J$ is simple. The same holds for $D_\psi$. So, there are at most
  $\tau(4n-2)$ intervals in which both $D_\sigma$ and $D_\psi$ are simple (and
  hyperbolic). In each interval $D_\sigma$ and $D_\psi$ intersect at most
  twice.
\end{proof}

We again observe that \c can only jump from one entity to another if they are
$\eps$-connected. Hence, Observation~\ref{obs:jump} holds entities moving in
$\R^d$ as well. As before, this means that at any time $t$, we can partition \X
into maximal sets of $\eps$-connected entities. Let $\X' \ni \sigma$ be a
maximal subset of $\eps$-connected entities at time $t$. This time, a central
trajectory \c uses the trajectory of entity $\sigma$ at time $t$, if and only
if $\sigma$ is the entity from $\X'$ whose function $\D_\sigma$ is
minimal. Hence, if we define $f_\sigma = D_\sigma$
Observation~\ref{obs:lower_envelope_1d} holds again as well.

Consider all $m'=O(n^2)$ intervals $J'_1,..,J'_{m'}$ that together form $[t_j,t_{j+1}]$. We
subdivide these intervals at points where the distance between two entities is
exactly $\eps$. Let $J_1,..,J_m$ denote the set of resulting intervals. Since
there are $O(n^2)$ times at which two entities are at distance exactly $\eps$,
we still have $O(n^2)$ intervals. Note that for all intervals $J_i$ and all
entities $\sigma$,  $f_\sigma$ is simple and totally defined on
$J_i$.

In each interval $J_i$, a central trajectory \c uses the trajectories of only
one maximal set of $\eps$-connected entities. Let $\X'_i$ be this set, let
$\F'_i = \{f_\sigma \mid \sigma \in \X'_i\}$ be the set of corresponding
functions, and let $\L_i$ be the lower envelope of $\F'_i$, restricted to
interval $J_i$. We now show that the total complexity of all these lower
envelopes is $O(n^{5/2})$. It follows that the maximal complexity of \c in
$J_i$ is at most $O(n^{5/2})$ as well.

\begin{lemma}
  \label{lem:complexity_L_in_interval}
  Let $J$ be an interval, let $\F$ be a set of hyperbolic functions that are
  simple and totally defined on $J$, and let $k$ denote the complexity of the
  lower envelope $\L$ of \F restricted to $J$. Then there are $\Omega(k^2)$
  intersections of functions in $\F$ that do not lie on $\L$.
\end{lemma}
\frank{It seems this should work for arbitrary functions that intersect at most
twice and span $J$. So we can generalize the lemma a bit further. Probably the
same holds for Lemma 24 then.}
\maarten{Or you could define it for arbitrary functions, and let $k$ be the number of distinct ones that appear on the lower envelope.}

\begin{wrapfigure}[14]{r}{0.50\textwidth}
    \centering
    \vspace{-1em}
    \includegraphics[scale=.8]{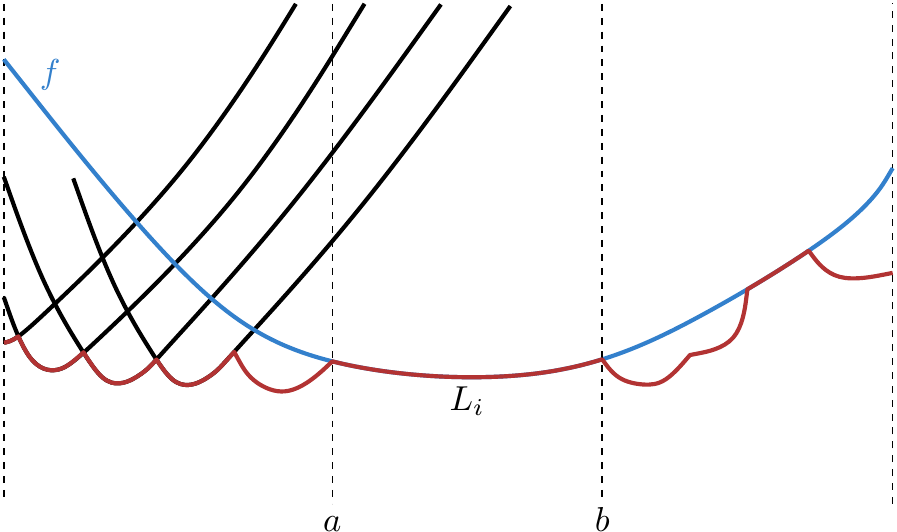}
    \caption{The function $f$ (blue) has at least $\ell_{i-2} = \lfloor (i-1)/2\rfloor$
      functions from $\F_i$ (black).}
    \label{fig:intersections_interval}
    \vspace{-1em}
\end{wrapfigure}
\begin{proof}
  Let $\L = L_1,..,L_k$ denote the pieces of the lower envelope, ordered from
  left to right. Consider any subsequence $\L'=L_1,...,L_i$ of the pieces. The
  functions in \F are all hyperbolic, so every pair of functions intersect at
  most twice. Therefore $\L'$ consists of at most $\lambda_2(|\F|) = 2|\F|-1$
  pieces. Hence, $i \leq 2|\F|-1$. The same argument gives us that there must
  be at least $\ell_i = \lfloor (i+1)/2 \rfloor$ distinct functions of $\F$
  contributing to $\L'$.

  Consider a piece $L_i = [a,b]$ such that $a$ is the first time that a
  function $f$ contributes to the lower envelope. That is, $a$ is the first
  time such that $f(t) = \L(t)$. Clearly, there are at least $\ell_k$ such
  pieces. Furthermore, there are at least $\ell_{i-2}$ distinct functions
  corresponding to the pieces $L_1,..,L_{i-2}$. Let $\F_i$ denote the set of
  those functions.

  All functions in $\F$ are continuous and totally defined, so they span time
  interval $J$. It follows that all functions in $\F_i$ must intersect $f$ at
  some time after the start of interval $J$, and before time $a$. Since $a$ was the
  first time that $f$ lies on \L, all these intersection points do not lie on
  \L. See Fig.~\ref{fig:intersections_interval}. In total we have at least
  \[     \sum_{i=2}^{\ell_k} \ell_{i-2}
    =    \sum_{i=2}^{\lfloor (k+1)/2 \rfloor} \lfloor (i-1)/2 \rfloor
    =    \sum_{i=1}^{\lfloor (k+1)/2 \rfloor -1} \lfloor i/2 \rfloor
    % \geq \frac{1}{2}\sum_{i=1}^{\lfloor k / 2 \rfloor - 1} (i-3)
    % =    \frac{1}{2}\sum_{i=1}^{\lfloor k / 2 \rfloor - 1} i - 3\lfloor k /2 -1\rfloor
    =    \Omega(k^2)
  \]
  such intersections.
\end{proof}

\begin{lemma}
  \label{lem:summed_lower_envelope_complexity}
  Let $\F_1,..,\F_m$ be a collection of $m$ sets of hyperbolic partial
  functions, let $J_1,..,J_m$ be a collection of intervals such that:
  \begin{itemize}[nosep]
  \item the total number of intersections between functions in $\F_1,..,\F_m$
    is at most $O(n^3)$,
  \item for any two intersecting intervals $J_i$ and $J_j$, $\F_i$ and $\F_j$
    are disjoint, and
  \item for every set $\F_i$, all functions in $\F_i$ are simple and totally
    defined on $J_i$.
  \end{itemize}
  Let $\L_i$ denote the lower envelope of $\F_i$ restricted to
  $J_i$,
  The total complexity of the lower envelopes $\L_1,..,\L_m$ is $O((m +
  n^2)\sqrt{n})$.
\end{lemma}

\begin{proof}
  Let $k_i$ denote the complexity of the lower envelope $\L_i$. An interval
  $J_i$ is \emph{heavy} if $k_i > \sqrt{n}$ and \emph{light}
  otherwise. Clearly, the total complexity of all light intervals is at most
  $O(m\sqrt{n})$. What remains is to bound the complexity of all heavy
  intervals.

  Relabel the intervals such that $J_1,..,J_h$ are the heavy intervals. By
  Lemma~\ref{lem:complexity_L_in_interval} we have that in each interval $J_i$,
  there are at least $ck^2_i$ intersections involving the functions $\F_i$, for
  some $c \in \R$.

  Since for every pair of intervals $J_i$ and $J_j$ that overlap the sets
  $\F_i$ and $\F_j$ are disjoint, we can associate each intersection with at
  most one interval. There are at most $O(n^3)$ intersections in total, thus we
  have $c'n^3 \geq \sum_{i=1}^m ck^2_i \geq \sum_{i=1}^h ck^2_i$, for some $c'
  \in \R$. Using that for all heavy intervals $k_i > \sqrt{n}$ we obtain
  \[      c'n^3
     \geq \sum_{i=1}^h ck^2_i
     \geq \sum_{i=1}^h c\sqrt{n}k_i
     =    c\sqrt{n} \sum_{i=1}^h k_i.
  \]

  It follows that the total complexity of the heavy intervals is $\sum_{i=1}^h
  k_i \leq c'n^3/c\sqrt{n} = O(n^2\sqrt{n})$.
\end{proof}

By Lemma~\ref{lem:intersections_F} we have that the number of intersections
between functions in \F in time interval $[t_j,t_{j+1}]$ is $O(n^3)$. Hence,
the total number of intersections over all functions in all sets $\F'_i$ is
also $O(n^3)$. All functions in each set $\F'_i$ are simple and totally defined
on $J_i$, and all intervals $J_1,..,J_m$ are pairwise disjoint, so we can use
Lemma~\ref{lem:summed_lower_envelope_complexity}. It follows that the total
complexity of $\L'_1,..,\L'_m$ is at most $O(n^{5/2})$. Thus, in a single time
interval the worst case complexity of \c is also at most $O(n^{5/2})$. We conclude:

\begin{theorem}
  \label{thm:complexity_c_d}
  Given a set of $n$ trajectories in $\R^d$, each with vertices at times
  $t_0,..,t_\tau$, a central trajectory \c has worst case complexity
  $O(\tau n^{5/2})$.
\end{theorem}

\subsection{Algorithm}
\label{sub:Algorithm_higher_dim}

We use the same global approach as before: we represent a set of
trajectoids containing an optimal solution by a graph, and then compute a
minimum weight path in this graph.
The graph that we use, is a slightly modified Reeb graph. We split an edge
$e$ into two edges at time $t$ if there is an entity $\sigma \in C_e$ such that
$D_\sigma = f_\sigma$ has a break point at time $t$. All functions $f_\sigma$,
with $\sigma \in C_e$, are now simple and totally defined on $J_e$. This
process adds a total of $O(\tau n^2)$ degree-two vertices to the Reeb
graph. Let \RG denote the resulting Reeb graph (see Fig.~\ref{fig:overview_2d}).

\begin{figure}[t]
  \centering
  \includegraphics[width=\textwidth]{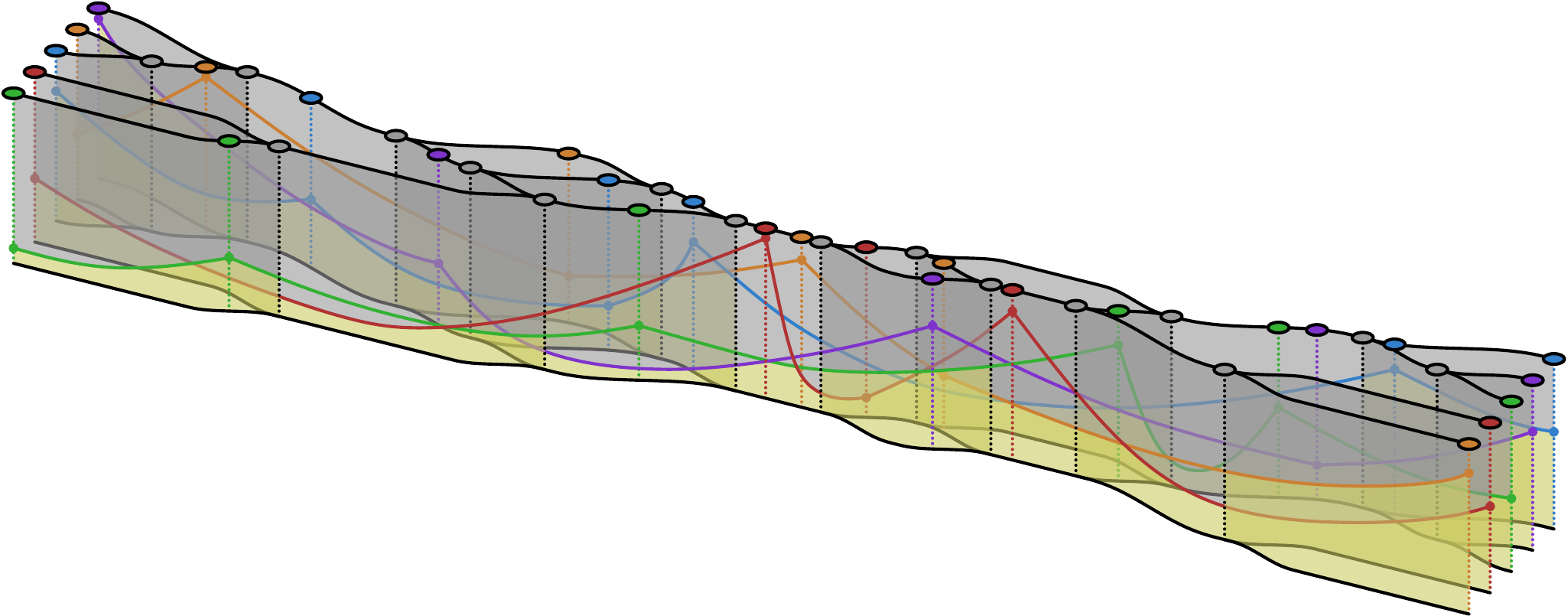}
  \caption{The modified Reeb graph \RG for five entities moving during a single
    interval, and the corresponding functions $f_\sigma$ for each entity $\sigma$.}
  \label{fig:overview_2d}
\end{figure}

To find all the times where we have to insert vertices, we explicitly compute
the functions $D_\sigma$. This takes $O(\tau n \lambda_2(n)\log n) = O(\tau n^2
\log n)$ time, where $\lambda_s$ denotes the maximum complexity of a
Davenport-Schinzel sequence of order $s$~\cite{as-dssga-00}, since within each
time interval $[t_i,t_{i+1}]$ each $D_\sigma$ is the upper envelope of a set of
$n$ functions that intersect each other at most twice. After we sort these
break points in $O(\tau n^2 \log n)$ time, we can compute the modified Reeb
graph \RG in $O(\tau n^2 \log n)$ time ~\cite{grouping2013}.\footnote{The
  algorithm to compute the Reeb graph is presented for entities moving in
  $\R^2$ in~\cite{grouping2013}, but it can easily be extended to entities
  moving in $\R^d$.}

Next, we compute all weights $\omega_e$, for each edge $e$. This means we have
to compute the lower envelope $\L_e$ of the functions $\F_e = \{ f_\sigma \mid
\sigma \in C_e \}$ on the interval $J_e$. All these lower envelopes have a
total complexity of at most $O(\tau n^{5/2})$:

\begin{lemma}
  \label{lem:complexity_Fes}
  The total complexity of the lower envelopes for all edges of the
  Reeb graph is $O(\tau n^{5/2})$.
\end{lemma}

\begin{proof}
  We consider each time interval $J_i=[t_i,t_{i+1}]$ separately. Let $\RG_i$
  denote the Reeb graph restricted to $J_i$. We now show that for each $\RG_i$,
  the total complexity of all lower envelopes $\L_e$ of edges $e$ in $\RG_i$ is
  $O(n^2\sqrt{n})$. The lemma then follows.

  By Lemma~\ref{lem:intersections_F}, the total number of intersections of all
  functions $\F_e$, with $e$ in $\RG_i$, is $O(n^3)$. Each set $\F_e$
  corresponds to an interval $J_e$, on which all functions in $\F_e$ are simple
  and totally defined. Furthermore, at any time, every entity is in at most one
  component $C_e$. So, if two intervals $J_e$ and $J_{e'}$ overlap, the sets of
  entities $C_e$ and $C_{e'}$, and thus also the sets of functions $\F_e$ and
  $\F_{e'}$ are disjoint. It follows that we can apply
  Lemma~\ref{lem:summed_lower_envelope_complexity}. Since $\RG_i$ has
  $O(n^2)$ edges the total complexity of all lower envelopes is
  $O(n^2\sqrt{n})$.
\end{proof}

We again have two options to compute all lower envelopes: either we compute all
of them from scratch in $O(\tau n^2 \cdot \lambda_2(n) \log n) = O(\tau n^3
\log n)$ time, or we use a similar approach as before. For each time interval,
we compute the arrangement \A of all functions \F, and then trace $\L_e$ in \A
for every edge $e$. For $n^2$ functions that pairwise intersect at most twice,
the arrangement can be built in $O(n^2 \log n + A)$ expected time, where $A$ is
the output complexity \cite{as-aa-00}. The complexity of \A is $O(n^3)$, so we
can construct it in $O(n^3)$ time. As before, every edge is traversed at most
once so tracing all lower envelopes $\L_e$ takes $O(n^3)$ time. It follows that
we can compute all edge weights in $O(\tau n^3)$ time, using $O(n^3)$ working
space.

Computing a minimum weight path takes $O(\tau n^2)$ time, and uses $O(\tau
n^2)$ space as before. Thus, we can compute \c in $O(\tau n^3)$ time and $O(n^3
+ \tau n^2)$ space.

\paragraph{Reducing the required working space} We can reduce the amount of
working space required to $O(n^2 \log n + \tau n^2)$ as follows. Consider
computing the edge weights in the time interval $J=[t_i,t_{i+1}]$. Interval $J$
is subdivided into $O(n^2)$ smaller intervals $J_1,..,J_m$ as described in
Section~\ref{sub:Complexity_2D}. We now consider groups of $r$ consecutive
intervals. Let $J$ be the union of $r$ consecutive intervals, we compute
the arrangement \A of the functions \F, restricted to time interval $J$. Since
every interval $J_i$ has at most $O(n^2)$ intersections \A has worst case
complexity $O(rn^2)$. Thus, at any time we need at most $O(rn^2)$ space to
store the arrangement. In total this takes $O(\sum_{i=1}^{n^2/r} (n_i \log n_i
+ A_i))$ time, where $n_i$ is the number of functions in the $i^\mathrm{th}$
group of intervals, and $A_i$ is the complexity of the arrangement in group
$i$. The total complexity of all arrangements is again $O(n^3)$. Since we cut
each function $D_\sigma$ into an additional $O(n^2/r)$ pieces, the total number
of functions is $O(n^3/r + n^2)$. Hence, the total running time is
$O((n^3/r)\log n + n^3)$. We now choose $r=\Theta(\log n)$ to compute all edge
weights in $[t_i,t_{i+1}]$ in $O(n^3)$ time and $O(n^2\log n)$ space. We
conclude:

% \frank{ For if we can also do the partial graph storage: }
% \paragraph{Reducing the required working space} We can reduce the amount of
% working space required to $O(n^2 \log n)$ by the following two modifications.

% Firstly, consider computing the edge weights in the time interval
% ...

% Secondly, we can use the same approach as in $\R^1$ to avoid storing the entire
% graph. Instead, it suffices to store only $O(n)$ values at any time.

\begin{theorem}
  \label{thm:central_trajectory_2d}
  Given a set of $n$ trajectories in $\R^d$, each with vertices at times
  $t_0,..,t_\tau$, we can compute a central trajectory \c in $O(\tau n^3)$ time
  using $O(n^2 \log n + \tau n^2)$ space. % \frank{+ $k$ space for the output complexity?}
  % \maarten {Hmm... good question. I guess we don't need to add it, since we do not necessarily need to store the output (maybe it goes into some communication pipe), and it is obvious if we do...}
\end{theorem}

\section{Extensions}
\label{sec:Extensions}

We now briefly discuss how our results can be extended in various directions.

\paragraph {Other measures of centrality}

We based our central trajectory on the center of the smallest enclosing disk of
a set of points.  Instead, we could choose other static measures of centrality,
such as the Fermat-Weber point, which minimizes the sum of distances to the
other points, or the center of mass, which minimizes the sum of squared
distances to the other points. In both cases we can use the same general
approach as described in Section~\ref{sec:higher_dimensions}.

Let $\hat{D}^2_\sigma(t) = \sum_{\psi \in \X} \|\sigma(t)\psi(t)\|^2$ denote
the sum of the squared Euclidean distances from $\sigma$ to all other entities
at time $t$. This function $\hat{D}^2_\sigma$ is piecewise quadratic in $t$,
and consists of (only) $O(\tau)$ pieces. It follows that the total number of
intersections between all functions $\hat{D}^2_\sigma$, $\sigma \in \X$, is at
most $O(\tau n^2)$. We again split the domain of these functions into
elementary intervals. The Reeb graph \RG representing the $\eps$-connectivity
of the entities still has $O(\tau n^2)$ vertices and edges. Each vertex of \c
corresponds either to an intersection between two functions $\hat{D}^2_\sigma$
and $\hat{D}^2_\psi$, or to a jump, occurring at a vertex of \RG. It now
follows that \c has complexity $O(\tau n^2)$.

To compute a central trajectory, we compute a shortest path in the
(weighted) Reeb graph \RG. To compute the weights we again construct the
arrangement of all curves $\hat{D}^2_\sigma$, and trace the lower envelope
$\L_e$ of the curves associated to each edge $e \in \RG$. This can be done in $O(\tau n^2 \log n)$ time
in total.

The sum of Euclidean distances $\hat{D}_\sigma(t) = \sum_{\psi \in \X}
\|\sigma(t)\psi(t)\|$ is a sum of square roots, and cannot be represented analytically in an efficient
manner. Hence, we cannot efficiently compute a central trajectory for this measure.

\frank{What about lower bounds for all these variants?}

Similarly, depending on the application, we may prefer a different way of integrating over time. Instead of the integral of $D$, we may, for example, wish to minimize $\max_t D(\cdot,t)$ or $\int D^2(\cdot,t)$.
Again, the same general approach still works, but now, after constructing the Reeb graph, we compute the weights of each edge differently.

\frank{$\max$ should be equally expensive as $\int$: pwhyp. so maximum at break
point. Hence, we need to evaluate the lower envelopes at every break point.}

\frank{using $D^2$ should also just be equally expensive I guess.}
%\maarten {Anything become easier here?}

%\maarten {theorem?}

\paragraph{Minimizing the distance to the Ideal Trajectory \I} We saw that for
entities moving in $\R^1$, minimizing the distance from \c to the farthest
entity is identical to minimizing the distance from \c to the ideal trajectory
\I (which itself minimizes the distance to the farthest entity, but is not
constrained to lie on an input trajectory). We also saw that for entities
moving in $\R^d$, $d > 1$, these two problems are not the same. So, a natural
question is if we can also minimize the distance to \I in this case. It turns
out that, at least for $\R^2$, we can again use our general approach, albeit
with different complexities.

Demaine \etal~\cite{demaine2010kinetic} show that for entities moving along
lines\footnote{Or, more generally, along a curve described by a low degree
  polynomial.} in $\R^2$ the ideal trajectory \I has complexity
$O(n^{3+\delta})$ for any $\delta > 0$. \frank{What about $\R^d$?} It follows that the function
$\check{D}_\sigma(t) = \|\I(t)\sigma(t)\|$ is a piecewise hyperbolic function
with at most $O(\tau n^{3+\delta})$ pieces. The total number of intersections
between all functions $\check{D}_\sigma$, for $\sigma \in \X$, is then $O(\tau
n^{5+\delta})$. Similar to Lemma~\ref{lem:summed_lower_envelope_complexity}, we
can then show that all lower envelopes in \RG together have complexity $O(\tau
n^{4+\delta})$.
% (essentially, by replacing $n^3$ by $n^{5+\delta}$, and $\sqrt{n}$ by
% $n^{1+\delta}$ in the proof of this lemma)
We then also obtain an $O(\tau n^{4+\delta})$ bound on the complexity of a
central trajectory \c minimizing the distance to \I.

To compute such a central trajectory \c we again construct \RG. To compute the
edge weights it is now more efficient to recompute the lower envelope $\L_e$
for each edge $e$ from scratch. This takes $O(\tau n^3 \cdot n \log n) = O(\tau
n^4 \log n)$ time, whereas constructing the entire arrangement may take up to
$O(\tau n^{5+\delta})$ time.

We note that the $O(n^{3+\delta})$ bound on the complexity of \I by Demaine
\etal~\cite{demaine2010kinetic} is not known to be tight. The best known lower bound is
only $\Omega(n^2)$. So, a better upper bound for this problem immediately also
gives a better bound on the complexity of \c.

\paragraph {Relaxing the input pieces requirement} We require each piece of the
central trajectory to be part of one of the input trajectories, and allow small
jumps between the trajectories. This is necessary, because in general no two
trajectories may intersect. Another interpretation of this dilemma is to
relax the requirement that the
output trajectory stays on an input trajectory at all times, and just require
it to be \emph {close} (within distance $\eps$) to an input trajectory at all times.
In this case, no discontinuities in the output trajectory are necessary.

% \frank{I removed the following from the above paragraph, since using the
% tubes apparently does not prevent jumps.}
% As a result, the
% output trajectory may be discontinuous. If we do not allow jumps and consider
% entities moving in $\R^d$, $d > 1$, then in general there is no possibility
% to switch between trajectories. Alternatively,

We can model this by replacing each point entity by a disk of radius $\eps$. The
goal then is to compute a shortest path within the union of disks at all times.
 We now observe that
if at time $t$ the ideal trajectory \I is contained in the same component of
$\eps$-disks as \c, the central trajectory will follow \I. If \I lies outside
of the component, \c travels on the border of the $\eps$-disk (in the component
containing \c) minimizing $D(\cdot,t)$. In terms of the distance functions,
this behavior again corresponds to following the lower envelope of a set of
functions. We can thus identify the following types of break points of \c: (i)
break points of \I, (ii) breakpoints in one of the lower envelopes
$\L_1,..,\L_m$ corresponding to the distance functions of the entities in each
component, and (iii) break points at which \c switches between following \I and
following a lower envelope $\L_j$. There are at most $O(\tau n^{3+\delta})$
break points of type (i)~\cite{demaine2010kinetic}, and at most $O(\tau
n^2\sqrt{n})$ of type (ii). The break points of type (iii) correspond to
intersections between \I and the manifold that we get by tracing the
$\eps$-disks over the trajectory. The number of such intersections is at most
$O(\tau n^{4+\delta})$. Hence, in this case \c has complexity $O(\tau
n^{4+\delta})$. We can thus get an $O(\tau n^{5+\delta} \log n)$ algorithm by
computing the lower envelopes from scratch.

%\clearpage
\small
\section*{Acknowledgments}
M.L. and F.S. are supported by the Netherlands Organisation for Scientific
Research (NWO) under grant 639.021.123 and 612.001.022, respectively.

%---------- Include Bibliography ---

% \bibliographystyle{abbrvnat}
%\bibliographystyle{abbrv} % in the end replace abbrvnat with abbrv

\bibliographystyle{abbrv}
\bibliography{central_trajectories}

\clearpage
\appendix

\section{A continuous central trajectory for entities moving in $\R^1$}
\label{app:oned_no_jumps}

\subsection{Complexity}
\label{sub:Complexity}

We analyze the complexity of central trajectory \c; a \trajectoid that
minimizes $\D'$, in the case $\eps = 0$. We now show that on each elementary
interval \c has complexity $24n$. It follows that the complexity of \c during a
time interval $[t_i,t_{i+1}]$ is $24n^2+48n$, and that the total complexity is
$O(\tau n^2)$.

We are given $n$ lines representing the movement of all the entities during an
elementary interval. We split each line into two half lines at the point where
it intersects \I (the $x$-axis). This gives us an arrangement \A of $2n$
half-lines. A half-line $\ell$ is \emph{positive} if it lies above the
$x$-axis, and \emph{negative} otherwise. If the slope of $\ell$ is positive,
$\ell$ is \emph{increasing}. Otherwise it is \emph{decreasing}. For a given
point $p$ on $\ell$, we denote the ``sub''half-line of $\ell$ starting at $p$
by $\ell^{p\to}$.
Let \c be a \trajectoid that minimizes $\D'$.

\begin{lemma}
  \label{lem:opt_clip_half_lines_at_v}
  Let $\ell$ and $m$ be two positive increasing half-lines, of which $\ell$ has
  the largest slope, and let $v$ be their intersection point. At vertex $v$, \c does
  not continue along $\ell^{v\to}$.
\end{lemma}

\begin{proof}
  Assume for contradiction that \c starts to travel along $\ell^{v\to}$ at
  vertex $v$ (see Fig.~\ref{fig:opt_clip_half_lines}). Let $t^*$, with $t^* >
  t_v$, be the first time where \c intersects $m$ again after visiting
  $\ell^{v\to}$, or $\infty$ if no such time exists. Now consider the
  \trajectoid \t, such that $\t(t) = m(t)$ for all times $t \in
  [t_v,t^*]$, and $\t(t) = \c(t)$ for all other times $t$. At any time $t$
  in the interval $(t_v,t^*)$ we have $0 \leq \t(t) < \c(t)$. It follows that
  $\D'(\t) < \D'(\c)$. Contradiction.
  \begin{figure}[htb]
    \centering
    \includegraphics{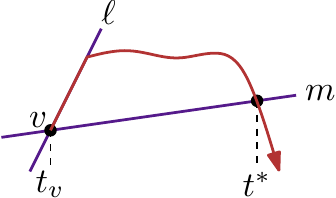}
    \caption{A central trajectory does not visit the half-line $\ell^{v\to}$.}
    \label{fig:opt_clip_half_lines}
  \end{figure}
\end{proof}

\begin{corollary}
  \label{cor:opt_clip_half_lines}
  Let $\ell$ and $m$ be two positive increasing half-lines, of which $\ell$ is
  the steepest, and let $v$ be their intersection point. \c does not visit
  $\ell^{v\to}$, that is, $\ell^{v\to} \cap \c = \emptyset$.
\end{corollary}

Consider two positive increasing half-lines $\ell$ and $\ell'$ in \A, of which
$\ell$ is the steepest, and let $v$ be their intersection
point. Corollary~\ref{cor:opt_clip_half_lines} guarantees that \c never uses
the half-line $\ell^{v\to}$. Hence, we can remove it from \A (thus replacing
$\ell$ by a line segment) without affecting \c. By
symmetry, it follows that we can \emph{clip} one half-line from every pair of
half-lines that are both increasing or decreasing. Let $\A'$ be the arrangement
that we obtain this way. See Fig.~\ref{fig:clipped_arrangement} for an
illustration.

\begin{figure}[h]
  % \centering
  \begin{subfigure}
    \centering
    \includegraphics[page=1]{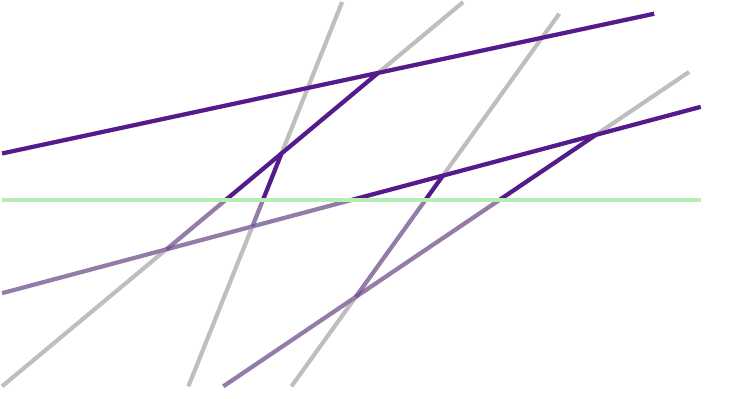}
  \end{subfigure}
  \begin{subfigure}
    \centering
    \includegraphics[page=2]{cut_arrangement}
  \end{subfigure}
  \caption{(a) The clipped increasing half-lines (purple) on top of the original
    half-lines (gray), and (b) the resulting arrangement $\A'$ and its zone (in
    yellow).}
  \label{fig:clipped_arrangement}
\end{figure}

Consider the set \mkmcal{Z} of open faces of $\A'$ intersected by \I, and let
$E$ be the set of edges bounding them. We refer to $\mkmcal{Z} \cup E$ as the
\emph{zone} of \I in $\A'$.

\begin{lemma}
  \label{lem:opt_in_zone}
  \c is contained in the zone of \I in $\A'$.
\end{lemma}

\begin{proof}
  We can basically use the argument as in
  Lemma~\ref{lem:opt_clip_half_lines_at_v}. Assume by contradiction that \c
  lies outside the zone from $u$ to $v$. The path from $u$ to $v$ along the
  border of the zone is $x$-monotone, hence there is a \trajectoid \t that
  follows this path during $[t_u,t_v]$ and for which $\t(t) = \c(t)$ at any
  other time. It follows that $\D'(\t) < \D(\c)$ giving us a contradiction.
\end{proof}

\newcommand{\Zl}{\ensuremath{\mkmcal{Z}_\ell}\xspace}
\begin{lemma}
  \label{lem:zonecomplex_Ap}
  The zone \Zl of a line $\ell$ in $\A'$ has maximum complexity $8n$.
\end{lemma}

\begin{proof}
  We show that the complexity of \Zl restricted to the positive half-plane is
  at most $4n$. A symmetric argument holds for the negative half-plane, thus
  proving the lemma. Since we restrict ourselves to the positive half-plane,
  the half-lines and segments of $\A'$ correspond to two forests: a purple
  forest with $p$ segments\footnote{Note that some of these segments ---the
    segments corresponding to the roots of the trees--- are actually
    half-lines.}, and a brown forest with $b$ segments. Furthermore, we have
  $p+b = n$.

  Rotate all segments such that $\ell$ is horizontal. We now show that the
  number of \emph{left-bounding} edges in \Zl is $2n$. Similarly, the number of
  right-bounding edges is also $2n$. Consider just the purple forest. Clearly,
  there are at most $p$ left-bounding edges in the zone of $\ell$ in the purple
  forest. We now iteratively add the edges of the brown forest in some order
  that maintains the following invariant: the already-inserted brown segments
  form a forest in which every tree is rooted at an unbounded segment (a
  half-line).  We then show that every new left-bounding edge in the zone
  either replaces an old left-bounding edge or can be charged to a purple
  vertex or a brown segment. In total we gather $p + 2b$ charges, giving us a
  total of $2p+2b = 2n$ left bounding edges.

  Let $s$ be a new brown leaf segment that we add to $\A'$, and consider the
  set $J$ of all intersection points of $s$ with edges that form \Zl in the
  arrangement so far. The points in $J$ subdivide $s$ into subsegments
  $s_1,..,s_k$ (See Fig.~\ref{fig:subsegments_in_zone}). All new edges in \Zl
  are subsegments of $s$. We charge the subsegment $s_i$ that intersects $\ell$
  (if any), and $s_k$ to $s$ itself. The remaining subsegments replace either
  a brown edge or a purple vertex from \Zl, or they yield no new left bounding
  edges.
  Clearly, segments
  replacing edges on \Zl do not increase the complexity of \Zl. We charge the
  segments replacing purple vertices to those vertices.
  We claim that a vertex $v$ gets charged at most once.
  Indeed, each vertex has three incident edges, only two of which may intersect $\ell$.
  The vertex gets charged when a brown segment intersects those two edges between $v$ and $\ell$. After this, $v$ is no longer part of \Zl in that face (though it may still be in \Zl in its other faces).
  It follows that the total number of charges is $p + 2b$.
  %The lemma follows.
  %
  \begin{figure}[h]
    \centering
    \includegraphics{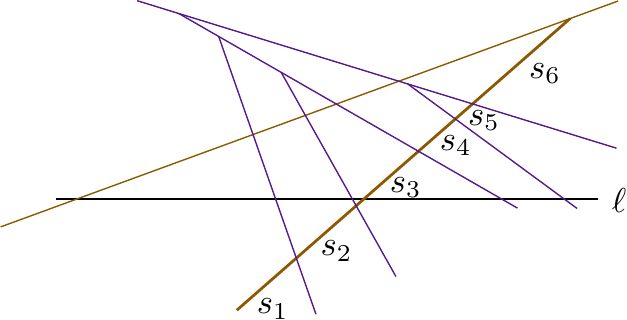}
    \caption{The new segment $s$ (fat) subdivided into subsegments. $s_3$ and
      $s_6$ are charged to $s$, $s_5$ is charged to a purple vertex. Segment
      $s_4$ replaces a purple edge in the zone, and $s_1$ and $s_2$ do not give
      new left bounding edges.}
    \label{fig:subsegments_in_zone}
  \end{figure}
\end{proof}

\begin{theorem}
  \label{thm:complexity_lines}
  Given $n$ lines, a trajectoid \c that minimizes $\D'$ has worst case
  complexity $8n$.
\end{theorem}

\begin{proof}
  \c is contained in the zone of $\A'$. So the intersection vertices of \c are
  vertices in the zone of $\A'$. By Lemma~\ref {lem:zonecomplex_Ap}, the zone has at most $8n$ vertices,
  %\frank{We should argue this I guess, since $\A'$ is no longer just an arrangement of lines}
     so \c has at most $8n$ vertices as well.
  % Every intersection vertex in \c corresponds to a vertex in $\A'$. Arrangement
  % $\A'$ is grid like, and \c follows an $x$-monotone path in this grid. It
  % follows that \c encounters at most $O(n)$ vertices of $\A'$. Hence \c
  % contains at $O(n)$ intersection vertices.\frank{TODO:be more specific what
  %   `gridlike' means}
\end{proof}

\eenplaatje [scale=0.8] {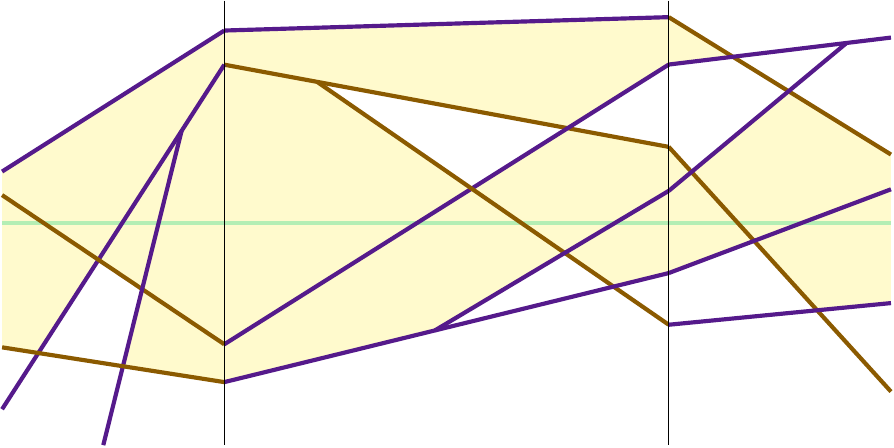} {A piece of the zone in three consecutive slabs.}

Let $\A^*$ be the total arrangement of all restricted functions; that is, a
concatenation of the arrangements $\A'$ restricted to the vertical slabs
defined by their elementary time intervals. Define the \emph {global zone} as the zone of $J$ in $\A^*$. Note that the global zone is more than just the union of the individual zones in the slabs, since cells can be connected along break points (and are not necesarily convex anymore). Nonetheless, we can still show that the complexity of the global zone is linear.

\begin {lemma}
  The global zone has complexity $24\tau n^2 + 48\tau n$.
\end {lemma}

\begin {proof}
The global zone is a subset of the union of the zones of $J$ and the vertical lines $x=t_i$, for $i \in {0,..,k}$ in the arrangements $\A'$.
By Lemma~\ref {lem:zonecomplex_Ap}, a line intersecting a single slab has zone complexity $8n$.
Each slab is bounded by two vertical lines and intersected by $J$, so applying the lemma three times yields a $24n$ upper bound on the complexity in a single slab.
Since there are $\tau(n+2)$ elementary intervals,
we conclude that the total complexity is at most $24 \tau n^2 + 48\tau n$.
\end {proof}

As before, it follows that \c is in the zone of \I in $\A^*$. Thus, we conclude:

\begin{theorem}
  \label{thm:complexity_opt_central trajectory}
  Given a set of $n$ trajectories in $\R^1$, each with vertices at times
  $t_0,..,t_\tau$, a central trajectory \c with $\eps = 0$, has worst case
  complexity $O(\tau n^2)$.
\end{theorem}

\subsection{Algorithm}
\label{sub:Algorithm}

We now present an algorithm to compute a \trajectoid \c minimizing $\D'$. It
follows from Lemma~\ref{lem:central_ideal} that such a trajectoid is a central
trajectory. The basic idea is to construct a weighted graph that represents a
set of \trajectoid{}s, and is known to contain an optimal \trajectoid. We then
find \c by computing a minimum weight path in this graph.

The graph that we use is simply a weighted version of the global zone of \I. We
augment each edge $e=(u,v)$ in the global zone with a weight $\int_{t_u}^{t_v}
|e(t)|\dd t$. Hence, we obtain a weighted graph \G. Finally, we add one source
vertex that we connect to the vertices at time $t_0$ with an edge of weight
zero, and one target vertex that we connect with all vertices at time
$t_\tau$. This graph represents a set of \trajectoid{}s, and contains an
optimal \trajectoid \c.

We find \c by computing a minimum weight path from the source to the target
vertex. All vertices except the source and target vertex have constant
degree. Furthermore, all zones have linear complexity. It follows that \G has
$O(\tau n^2)$ vertices and edges, and thus, if we have \G, we can compute \c in
$O(\tau n^2 \log (\tau n))$ time.

We compute \G by computing the zone(s) of the arrangement $\A'$ in each
elementary interval. We can find the zone of $\A'$ in $O((n+k)\alpha(n+k)\log
n) = O(n\alpha(n)\log n)$ expected time, where $k$ is the complexity of \c and
$\alpha$ is the inverse Ackermann function, using the algorithm of
Har-Peled~\cite{harpeled2000walk}. Since $\A'$ has a special shape, we can
improve on this slightly as follows.

We use a sweep line algorithm which sweeps $\A'$ with a vertical line from left
to right. We describe only computing the upper border of the zone (the part
that lies above \I). Computing the lower border is analogous. So in the
following we consider only positive half-lines.

We maintain two balanced binary search trees as status structures. One storing
all increasing half-lines, and one storing all decreasing half-lines. The
binary search trees are ordered on increasing $y$-coordinate where the
half-lines intersect the sweep line. We use a priority queue to store all
events. We distinguish the following events: (i) a half-line starts or stops at
\I, (ii) an increasing (decreasing) half-line stops (starts) because it
intersects an other increasing (decreasing) half-line, and (iii) we encounter
an intersection vertex between an increasing half-line and a decreasing
half-line that lies in the zone. In total there are $O(n)$ events.

The events of type (ii) involve only neighboring lines in the status structure,
and the events of type (iii) involve the lowest increasing (decreasing)
half-line and the decreasing (increasing) half-lines that are in the zone when
they are intersected by the sweep line. To maintain the status structures and
compute new events we need constantly many queries and updates to our status
structures and event queue. Hence, each event can be handled in $O(\log n)$
time.

The events of type (i) are known initially. The first events of type (ii) and
(iii) can be computed in $O(\log n)$ time per event (by inserting the
half-lines in the status structures). So, initializing our status structures
and event queue takes $O(n \log n)$ time. During the sweep we handle $O(n)$
events, each taking $O(\log n)$ time. Therefore, we can compute the zone of
$\A'$ in $O(n \log n)$ time in total.

\paragraph{Computing a minimum weight path} We can slightly improve the running
time by reducing the time required to compute a minimum weight path. If,
instead of a general graph $\G=(V,E)$ we have a directed acyclic graph (DAG),
we can compute am minimum weight path in only $O(|V|+|E|) = O(\tau n^2)$ time using
dynamic programming. We transform \G into a DAG by orienting all edges $e=(u,v)$, with $t_u <
t_v$, from $u$ to $v$.

% We can now also reduce the required working space. Consider a vertex $v$, and
% let $(u_1,v),..,(u_k,v)$ be its incoming edges. To compute the total weight
% $\D'$ up to vertex $v$, the dynamic programming algorithm requires only the
% total weight up to vertices $u_i$. At any time $t$, there are at most $O(n)$
% edges $e=(u,v)$ such that $t \in [t_u,t_v]$, hence at any time we have to store
% at most $O(n)$ total weight values. It follows that we do not have to
% completely construct \G before computing a minimum weight path.

The running time is now dominated by constructing the graph. We conclude:

\begin{theorem}
  \label{thm:central trajectory_1d_no_jumps}
  Given a set of $n$ trajectories in $\R^1$, each with vertices at times
  $t_0,..,t_\tau$, we can compute a central trajectory \c for $\eps=0$ in
  $O(\tau n^2 \log n)$ time using $O(\tau n^2)$ space.
\end{theorem}

\end{document}